\documentclass[aps,pre,amsmath,amsfonts,amssymb,superscriptaddress,showpacs,10pt]{revtex4}
\usepackage{graphics}
\usepackage[latin1]{inputenc}
\usepackage{amscd}
\usepackage{soul}

\usepackage[pdftex]{graphicx}

\usepackage{rotating}
\usepackage[pdftex,hypertexnames=true,linktocpage=true]{hyperref}

\usepackage{color}

\usepackage{amsthm}
\usepackage{framed}
\usepackage{amssymb}
\usepackage{amsmath,mathrsfs}
\usepackage{hhline}
\usepackage{ulem}
\usepackage{fancyhdr}
\usepackage{tikz}
\everymath{\displaystyle}
\everymath{\displaystyle}

\everymath{\displaystyle}
\newcommand{\beq}{\begin{equation}}
\newcommand{\eeq}{\end{equation}}

\newcommand{\bi}{\begin{itemize}}
\newcommand{\ei}{\end{itemize}}

\newcommand{\adj}{{\rm adj}}
\newcommand{\approach}{\ensuremath{\rightarrow}}
\renewcommand{\H}{\mathcal{H}}
\newcommand{\R}{\mathbb{R}}
\newcommand{\imply}{\ensuremath{\Rightarrow}}
\newcommand{\mapto}{\ensuremath{\rightarrow}}

\newcommand{\la}{\langle}
\newcommand{\ra}{\rangle}

\def\RR{\mathbb{R}}

\def\tr{{\rm{tr}}}

\newtheorem{definition}{Definition}[section]

\newtheorem{theorem}{Theorem}[section]    
\newtheorem{lemma}{Lemma}[section]   
\newtheorem{pro}{Proposition}[section]
    
\newtheorem{remark}{Remark}[section]

\newtheorem{corollary}{Corollary}
\newcommand{\trace}{{\rm tr}}

\begin{document}
\title{The volume of Gaussian states by information geometry}

\author{Domenico Felice}
\email{domenico.felice@unicam.it}
\affiliation{School of Science and Technology,
University of Camerino, I-62032 Camerino, Italy \\
INFN-Sezione di Perugia, Via A. Pascoli, I-06123 Perugia, Italy}
\author{H\`a Quang Minh}
\email{minh.haquang@iit.it}
\affiliation{Istituto Italiano di Tecnologia\\ 
Via Morego 30, I-16163 Genova , Italy}
\author{Stefano Mancini}
\email{stefano.mancini@unicam.it}
\affiliation{School of Science and Technology,
University of Camerino, I-62032 Camerino, Italy \\
INFN-Sezione di Perugia, Via A. Pascoli, I-06123 Perugia, Italy}

\begin{abstract}
We formulate the problem of determining the volume of the set of Gaussian physical states in the framework of information geometry. This is done by considering phase space probability distributions parametrized by their covariances and endowing the resulting statistical manifold with the Fisher-Rao metric. We then evaluate the volume of classical, quantum and quantum entangled states for two-mode systems, showing chains of strict inclusions.
\end{abstract}


\maketitle

\section{Introduction}

\noindent States of physical systems in classical and quantum mechanics are represented by very different mathematical structures, nevertheless analogies appear at certain points of comparison \cite{Mukunda}. Classical states are depicted as probability density functions (pdf$s$) in phase space, whereas quantum states are described by density operators defined on Hilbert spaces \cite{Blum}. In fact, the notion of phase space is often gotten rid of in quantum mechanics because of the non-commutativity of canonical variables. Nevertheless, phase space can be considered as common playground for both classical and quantum states when one employs for the latter a description in terms of the so-called quasi-probability distribution functions, such as the Wigner function \cite{Kim}.
Then, one can address the computation of the volume of different classes of states in the phase space framework. The issue of the volume of sets of states is of uppermost importance. It can help in distinguishing classical from quantum states as well as to find separable states within all quantum states. Separable states are the states of a composite system that can be written as convex combinations of subsystem states, in contrast to entangled states \cite{Horo}. Determining the volume of physical states is also relevant for defining ``typical'' properties of a set of states. In fact, to this end, one usually resorts to the random generation of states according to a suitable measure stemming from the volume of states \cite{Lupo}.

Describing the geometric properties of sets of states is intimately connected with the evaluation of their volumes. The sets of classical and quantum states are both convex sets. In finite dimensional systems, several metrics are introduced in order to compute the volume of physical states. Due to their own nature as pdf$s$, classical states can be distinguished by the well-known Fisher-Rao metric \cite{Naga}. Quantum analogue can be found in the setting of pure states, where the Fubini-Study metric turns out to be proportional to the Fisher-Rao metric \cite{Marmo1}. However, for quantum mixed states there is no single metric \cite{Cafaro}. Several measures have been analysed, each of them arising from different physical motivations and advantages \cite{GQ}. Such different measures have been proposed on the set of density matrices acting on a finite-dimensional Hilbert space; a very natural one employed the Positive Partial Transpose (PPT) criterion \cite{Horo} to determine an upper bound for the volume of separable quantum states and figuring out that, for any composite quantum system, it is different from zero regardless of the number of subsystems it contains and its (finite) dimension \cite{Peres}. Other important measures include the Hilbert-Schmidt measure, the Bures measure and the measure induced by partial trace on composite systems. All of them use techniques from geometric functional analysis and convex geometry to estimate the volume of separable quantum states \cite{Zycz,Aubrun,Szarek}. Finally, a generalization of the Hilbert-Schmidt measure and the volume induced by the partial trace on composite systems, the so-called $\alpha$-volume, showed that the PPT criterion is not precise for large dimensions of the Hilbert space \cite{Ye}. This is an evidence that with increasing Hilbert space dimensions, the procedure to test the separability becomes more and more difficult to implement. Additionally, when going to infinite dimensional systems (often referred to as `` continuous variable''  -CV- systems), problems also arise from the non-compactness of the support of states.

Thus, on the one hand, we have the difficulties in analysing infinite dimensional systems, while on the other hand we still lack a unifying approach for evaluating volumes of classical and quantum states. To deal with these problems, we propose to exploit Information Geometry. This is the application of differential geometric techniques to the study of families of probabilities \cite{Amari}. {As such, it can be applied to Gaussian states, be they either classical or quantum. Indeed, Gaussian classical states are pdf$s$ in phase space and the same is true for Gaussian quantum states, which are pdf$s$ coming from Wigner functions in phase space \cite{Mann}. The main reason for the focus on Gaussian states is that they are ubiquitous in physics, mathematics and information theory (see e.g. \cite{Mukunda}).}

Very recently, a method based on the extension of the Hilbert-Schmidt measure has been proposed \cite{Strunz} to evaluate the volume of Gaussian quantum states, which is not applicable however to the classical states. In the present work, we exploit methods of Information Geometry in order to associate a Riemannian manifold to a generic Gaussian system. In such a way, we consider a volume measure as the volume of the manifold associated to a set of states of the system. More specifically, we start by considering $N$ identical and indistinguishable particles, i.e. bosonic modes characterized by their positions and momenta and we assume that a Gaussian pdf with zero mean value describes the whole system state. Such a pdf is characterized by a set of parameters, i.e. the entries of the covariance matrix (depending on their values we can have various classes of states). Then, thanks to these parameters, to each class of states is associated a statistical model which turns out to be a Riemannian manifold endowed with the well-known Fisher-Rao metric (see also \cite{FMP}). {We are able to overcome the difficulty of an unbounded volume by introducing a regularizing function stemming from energy bounds, which acts as a form of compactification of the support of Gaussian states. We then proceed to consider a different regularizing function which satisfies some nice properties of canonical invariance. Finally, we find the volumes of classical, quantum, and quantum entangled states for two-mode Gaussian systems, showing chains of strict inclusions.}

The layout of the paper is as follows. In Sec. \ref{sec2} we recall the phase space representation of both classical and quantum states. Then, in Sec. \ref{sec3} we present a volume measure for Gaussian states based on information geometry. Sec. \ref{sec4} is devoted to the regularization of the introduced volume measure.  Applications to bipartite states of two mode systems are discussed in Sec. \ref{sec5}. Finally we draw our conclusions in Sec. \ref{sec6}.


\section{Phase space representation of states}\label{sec2}

The phase space $\Gamma$ of $N$ identical and indistinguishable particles (i.e. bosonic modes) is the $2N$-dimensional space of allowed real values for the canonical position and momentum variables $\xi=(q_1,p_1,\ldots,q_N,p_N)^T$ of such modes (by $T$ we denote the transpose). 

A classical state for such a system of $N$ modes is represented by a pdf in $\Gamma$, namely
\begin{eqnarray}\label{pdf}
\rho: \Gamma \to \mathbb{R}_+, \qquad \int_\Gamma d\xi\ \rho(\xi)=1.
\end{eqnarray}
As a particular case, when $\rho$ becomes a Dirac delta $\delta^{2N}(\xi-\xi_0)$ we have a pure state whose values of position and momentum variables are (deterministically) given by $\xi_0$. Throughout the paper we will consider $\Gamma=\RR ^{2N}$ and the integration is performed on  $\RR ^{2N}$  when not otherwise specified.

The probability density function in \eqref{pdf} can be considered as originating from the characteristic function 
$\chi_\rho(\tau)$, through the Fourier transform,
\begin{equation}
\label{charactfunctclassic}
\rho(\xi)=\int d\tau\ e^{-i\xi^T \tau} \chi_\rho(\tau),
\end{equation}
where $i$ is the imaginary unit and $\tau\in\RR ^{2N}$.

The set of all (mixed and pure) states is a convex set, that is if $\rho_j(\xi)$ for $j =1,2,\ldots$ represent states 
and $\left\{P_j\right\}_j$ is a probability vector, then
$$
\rho(\xi)=\sum_j  P_j \rho_j(\xi),
$$
is still a possible (mixed) state. 
Only pure states cannot be decomposed in a non trivial manner as convex sum of
other states, so they are the extremal points (or extremal elements) in the space of all states.

\bigskip

The quantum analogue of pdf is the density operator $\hat\rho$ defined on the Hilbert space ${\cal H}=L^2(\RR)^{\otimes N}$ associated to the $N$-mode system. The canonical position and momentum variables become operators $\hat{q}_k,\hat{p}_k,\ k=1,\ldots,N$ on ${\cal H}$ with the commutation relation $\left[\hat{q}_k,\hat{p}_k\right]=i$. Setting $\hat{R}_{2k-1}:=\hat{q}_k$ and $\hat{R}_{2k}:=\hat{p}_k$ these relations are summarized as $\left[\hat{R}_k,\hat{R}_l\right]=i\ \Omega_{kl}$, where $\Omega_{kl}$ is the $kl$ entry of the antisymmetric $2N\times 2N$ matrix 
\begin{equation}\label{J}
\Omega=\bigoplus_{j=1}^N\left(\begin{array}{cc}
0&1\\
-1&0\\
\end{array}\right).
\end{equation}
This induces a symplectic structure on the phase space $\Gamma$, meaning that a bilinear form $\omega:\Gamma\times \Gamma\rightarrow\RR$ exists, $\omega$ being non-degenerate and skew-symmetric.

A phase space representation of the state $\hat\rho$ can be given by means of the Wigner function 
defined as in \cite{Kim}
\begin{equation}\label{wigner}
W(\xi):=\left(\frac{1}{\pi}\right)^N\int\ e^{2i\sum_{k=1}^N p_ky_k}\rho(q_1+y_1,q_1-y_1,\ldots,q_N+y_N,q_N-y_N)dy_1\ldots dy_N.
\end{equation}
Here, $\rho(q_1+y_1,q_1-y_1,\ldots,q_N+y_N,q_N-y_N)$ is the position representation of the density operator, i.e. the representation of the operator $\hat\rho$ on the eigenvectors of the operators $\hat{q}_k$,  $k\in\{1,\ldots,N\}$. In such a way, the Wigner function turns out to be defined over the $2N$-dimensional phase space $\Gamma$ and would be the analogous of classical pdf$s$ $\rho(\xi)$. 
Nevertheless, the Wigner function is not a pdf because it can also assume negative values. Hence, it is often called a \textit{quasi-probability distribution function}. 

Yet, as a proper pdf, the Wigner function can be considered as originating from the characteristic function
$\chi_{\hat\rho}(\tau)$ through the Fourier transform,
\begin{equation}\label{charactfunctquantum}
W(\xi)=\int d\tau \ e^{-i\xi^T\tau}\ \chi_{\hat\rho}(\tau),
\end{equation}
where
\begin{equation}\label{characteristic}
\chi_{\hat\rho}(\xi):=\tr\left[\hat{\rho} \hat{D}(\xi)\right],
\end{equation}
and
\begin{equation*}
\hat{D}(\xi):=\exp\left[ i\sum_k\ \left(q_k \hat{q}_k + p_k \hat{p}_k \right)\right].
\end{equation*}


\subsection{Gaussian States}

Gaussian states are those for which the characteristic function is a Gaussian function of the phase space coordinates $\xi$, namely
\begin{equation}
\label{Gaussiancharct1}
\chi_\rho(\xi)=e^{-\frac{1}{4}\xi^T V\xi-i\mu^T\xi},
\end{equation}
or 
\begin{equation}
\label{Gaussiancharct2}
\chi_{\hat\rho}(\xi)=e^{-\frac{1}{4}\xi^T V\xi-i\mu^T\xi},
\end{equation}
where $V$ is the 2$N\times 2 N$ covariance matrix and $\mu\in\RR ^{2N}$ the first moment vector
(recall that a Gaussian state is completely determined by $V$ and $\mu$). 

Although formally identical, Eqs. \eqref{Gaussiancharct1} and \eqref{Gaussiancharct2} differ by the conditions imposed on the covariance matrix.
In fact, for classical states $V$ is symmetric and strictly positive definite, i.e. $V>0$.
Yet, not all symmetric, positive definite matrices correspond to the covariance matrices of quantum physical states. 
In fact, due to the non-commutativity of canonical operators we have \cite{uncertainty,Simon94}
\begin{theorem}\label{QSth}
A \textit{real}, symmetric $2 N\times 2N$ matrix $V>0$ describes a Gaussian quantum state if and only if\begin{equation}
\label{QS}
V+i\Omega\geq 0.
\end{equation}
\end{theorem}
{Relation \eqref{QS} is equivalent to Schr\"odinger's formulation of Heisenberg's uncertainty principle \cite{Narco90}. Furthermore, the {\textit{if and only if}} in Theorem \ref{QSth} is a peculiarity of Gaussian states as can be seen e.g. by  Hardy's formulation of Heisenberg's uncertainty principle \cite{deGosson07,Mukunda95}. {In any case},} it is trivial to show that { Eq. \eqref{QS}} implies the positive definiteness of the matrix $V$, { whereas} the converse is not true.  

The Gaussian form of the characteristic functions \eqref{Gaussiancharct1} and \eqref{Gaussiancharct2} reflects on the corresponding phase space representations $\rho(\xi)$ and $W(\xi)$ by Eqs.\eqref{charactfunctclassic} and \eqref{charactfunctquantum}, which we can commonly write as
\begin{eqnarray}\label{PxT}
P(\xi)=\frac{e^{-\frac{1}{2}\xi^T V^{-1}\xi}}{(2\pi)^N\sqrt{ \det V}}.
\end{eqnarray}
Here we set $\mu=0$ since the first moments are irrelevant for most of the physical properties of Gaussian states.
Notice that the Wigner function, being in such a case a Gaussian function, is a true pdf.

\bigskip

Among quantum states we can also distinguish between \textit{separable} and \textit{entangled} states \cite{Horo}.
To this end, it would be helpful to employ the partial transposition. It follows from the definition of the Wigner function \eqref{wigner} that on phase space, transposition corresponds to the transformation that changes the sign to all $p$s coordinates and leaves the $q$s unchanged
\begin{equation}
\label{transposition}
\left(q_1,p_1,\ldots,q_N,p_N\right)\mapsto\Lambda\left(q_1,p_1,\ldots,q_N,p_N\right):=\left(q_1,-p_1,\ldots,q_N,-p_N\right).
\end{equation}
Consider now a composite Gaussian system with two subsystems $A$ and $B$; let $V$ be the covariance matrix describing the whole system, and $V_A$ and $V_B$ be the ones describing subsystems $A$ and $B$, respectively. Denote by $\Lambda_A:=\Lambda\oplus {\rm id}$ (resp. $\Lambda_B= {\rm id}\oplus\Lambda$) the positive partial transposition in the $A$'s system only (resp. $B$'s system only). Then a necessary and sufficient condition for the separability of the system is given by the following theorem \cite{Werner}
\begin{theorem}\label{separabilityth}
A Gaussian state described by the covariance matrix $V$ is separable if and only if there exist covariance matrices $V_A$ and $V_B$ such that
\begin{equation}\label{separability}
V\geq V_A\oplus V_B.
\end{equation}
\end{theorem}
{Unfortunately this theorem is not easy to verify in practice since doing so requires looking for covariance matrices $V_A$ and $V_B$ satisfying \eqref{separability}.} 

Nevertheless, if we consider a two-mode Gaussian system ($N=2$), the criterion to distinguish separable from entangled states simplifies into \cite{Simon}
\begin{theorem}\label{sepcor}
The $4\times 4$ symmetric matrix $V$ satisfying the condition \eqref{QS} describes a separable state if and only if
\begin{equation}\label{separable}
\widetilde{V}+i\Omega\geq 0,
\end{equation}
where $\widetilde{V}=\Lambda_B V\Lambda_B$, with 
\begin{equation}
\label{PPT2}
\Lambda_B(q_1,p_1,q_2,p_2)=(q_1,p_1,q_2,-p_2).
\end{equation}
\end{theorem}


\section{A volume measure based on information geometry}\label{sec3}

\noindent Let us consider the family ${\cal S}$ of Gaussian pdf with zero mean in the $2N$-dimensional phase space $\Gamma$. Each such pdf takes on the form $P(\xi)$ of Eq. \eqref{PxT} and may be parametrized using $m\le N(2N+1)$ \textit{real}-valued variables $\theta^1,\ldots\theta^m$ (the nonzero entries of the covariance matrix) so that 
\begin{equation}\label{statmodel}
{\cal S}:=\left\{P(\xi)\equiv P(\xi;\theta)=\frac{e^{-\frac{1}{2}\xi^T V^{-1}(\theta)\xi}}{(2\pi)^N\sqrt{\det V(\theta)}},\ \Big| \ \theta\in\Theta \right\},
\end{equation}
where $\Theta$ is a subset of $\RR ^m$ obtained by requiring some specific constraints on $V(\theta)$, and the mapping $\theta \mapsto P(.;\theta)$ is injective. In such a way, ${\cal S}$ turns out to be an $m$-dimensional statistical model (in fact a Gaussian statistical model). The parametrization is provided by the entries of the covariance matrix $V=\left[V_{\mu\nu}\right]_{\mu\nu}$, by defining $\theta^l=V_{\mu\nu}$, with $l=\sum_{r=0}^{\mu-2}(2N - r) +\nu-\mu + 1$ and $1\leq l\leq m$ (the summation over $r$ disappears when $\mu=1$). 

{Given the statistical model ${\cal S}$ of Eq. \eqref{statmodel}, the mapping $\varphi:{\cal S}\rightarrow \RR^m$ defined by $\varphi(P(.;\theta)) = \theta$ is injective and it allows us to consider $\varphi=\left[\theta^l\right]$ as a coordinate system for ${\cal S}$. In addition, we assume that a change of coordinates $\psi:\Theta\rightarrow \psi(\Theta)\subset\RR^m$ is such that the set $\{P(.;{\psi^{-1}(\kappa)})\;\arrowvert\;\kappa\in\psi(\Theta)\}$, where 
$\kappa$ is the set of new coordinates given by $\kappa^l:=\psi(\theta^l)$, represents the same family of probability functions as ${\cal S}=\{P(.;\theta)\ \arrowvert\ \theta \in\Theta\}$. Moreover, we also assume that such a change of coordinates is differentiable. Thereby, ${\cal S}$ can be considered as a $C^\infty$ differentiable manifold, called \textit{statistical manifold} \cite{Amari}.}

\begin{remark}
\label{remprobdisj}
From here on, we assume that $\theta^{l+1}=0$ with $l=\sum_{r=0}^{\mu-2}(2N-r)+1$, for $\mu=2k+1$. This implies that there is no correlation between position $q_k$ and momentum $p_k$ of the $k$th mode, for all $k=1,\ldots,N$.
\end{remark}

Consider now a point $\theta\in\Theta$; then, the \textit{Fisher information matrix} of ${\cal S}$ at $ \theta$ is the $m\times m$ matrix $g(\theta)$ whose entries are given by \cite{Amari}
\begin{equation}
\label{gFR}
g_{\mu\nu}(\theta):=\int_{\RR^{2N}}dx\,P(\xi;\theta)\,\partial_\mu\ln P(\xi;\theta)\partial_\nu\ln P(\xi;\theta),
\end{equation}
with $\partial_\mu$ standing for $\frac{\partial}{\partial \theta^\mu}$. The resulting matrix $g(\theta)$ is symmetric and positive semidefinite. Yet, we assume from now on that $g(\theta)$ is positive definite. In such a way, we can endow the parameter space $\Theta$ with a Riemannian metric, the \textit{Fisher-Rao} metric, given by $G(\theta):=\sum_{\mu\nu}\,g_{\mu\nu}(\theta)\, d\theta^\mu\otimes d\theta^\nu$, with $ g_{\mu\nu}(\theta)$ as in Eq. \eqref{gFR}. With this metric, the manifold ${\cal M}:=\left(\Theta,G(\theta)\right)$ becomes a Riemannian manifold.

The parameter space $\Theta$ in \eqref{statmodel} does not coincide, in general, with the whole linear space $\RR ^m$. The central issue is that by requiring that the covariance matrix $V(\theta)$ satisfies some specific conditions, $\Theta$ can represent different states of the physical system.

\begin{definition}
\label{defiphysical}
Consider the Gaussian statistical model ${\cal S}=\{P(\xi;\theta)\}$ in Eq. \eqref{statmodel}. Then, the classical states of the physical system are represented by the parameter space $\Theta$ given by
\begin{equation}\label{classicstates}
\Theta_{\mbox{\tiny classic}}:= \{\theta\in\RR^m \arrowvert V(\theta)>0 \}.
\end{equation}
The quantum states, instead, are represented by means of the following parameter space $\Theta$,
\begin{equation}\label{quantumstates}
\Theta_{\mbox{\tiny quantum}}:= \{\theta\in\RR^m \arrowvert V(\theta)+i\Omega\geq 0 \}.
\end{equation}
If the physical system is composed by two subsystems $A$ and $B$, then its separable states stand for
\begin{equation}\label{separablestates}
\Theta_{\mbox{\tiny separable}}:= \{\theta\in\RR^m \arrowvert V(\theta)\geq V_A\oplus V_B \}.
\end{equation}
Finally, in this case, the entangled states are given by
\begin{equation}\label{entangledstates}
\Theta_{\mbox{\tiny entangled}}:= \Theta_{\mbox{\tiny quantum}}-\Theta_{\mbox{\tiny separable}}.
\end{equation}
\end{definition} 

\begin{remark}\label{rempar}

Eq. \eqref{classicstates} represents all the possible classical states of the physical system described by the $2N$-dimensional phase space $\Gamma$. A labelling permutation $\sigma$ of the system's modes acts on the pdf $P(\xi;\theta)$ by a permutation congruence of the covariance matrix: $V(\theta)\rightarrow \Pi^T V(\theta) \Pi$, where $\Pi$ is the permutation matrix corresponding to $\sigma$. Now, $\Pi^T V(\theta) \Pi$ is still positive definite; so, the parameter space $\Theta_{\mbox{\tiny classic}}$ has a permutation invariant form.

Eq. \eqref{quantumstates} represents all the possible quantum states of the physical system described by the $2N$-dimensional phase space $\Gamma$. It is well-known that the uncertainty relation $V(\theta)+i\Omega\geq 0$ has a symplectic invariant form \cite{Simon}, i.e. given any symplectic matrix $S$, then also $S^TV(\theta)S+i\Omega\geq 0$ holds true. Thus, the parameter space $\Theta_{\mbox{\tiny quantum}}$ has also a symplectic invariant form in addition to the permutation one. 

In general, the definition given in Eq. \eqref{separablestates} is not operational; indeed, to use it, it is necessary to prove the existence of the matrices $V_A$ and $V_B$. However, when dealing with two-mode systems, such a criterion becomes useful in practice by providing a necessary and sufficient condition to distinguish separable states among all the quantum states.

Going on to multipartite systems, the task to describe separable states becomes harder and harder. In fact, a general criterion is still missing.
\end{remark}


\subsection{The Volume measure}

\noindent {From Def. \ref{defiphysical} we see that, each different set of Gaussian states is associated with a Riemannian manifold. Thus, a natural volume measure for a set of states is the volume of the associated manifold.}

\begin{definition}\label{defivol}
Consider a physical system of $N$ modes in a Gaussian state. Let $\Theta$ be the parameter space as in \eqref{defiphysical} and ${\cal M}=(\Theta,G(\theta))$ be the Riemannian manifold associated to the class of Gaussian states $\Theta$, with $G(\theta)$ being the Fisher-Rao metric. Then the volume of the physical states represented by $\Theta$ is 
\begin{equation}
\label{volume}
{\cal V}(V):=\int_{\Theta}\ d\theta \sqrt{\det g(\theta)},
\end{equation}
where $g(\theta)$ is the \text{real} symmetric matrix with entries given by \eqref{gFR}.
\end{definition} 

Given the formal definition of the Fisher-Rao metric tensor \eqref{gFR}, in order to apply it in practice, we consider a clearer analytical relation between the components of the metric $G(\theta)$ and the covariance matrix $V(\theta)$.   
\begin{theorem}\label{main}
The entries \eqref{gFR} of the Fisher-Rao metric are related to $V$ by
\begin{equation}
g_{\mu\nu}=\frac{1}{2}\tr \left[V^{-1}\ \left(\partial_\mu V\right)\ V^{-1}\ \left(\partial_\nu V\right)\right],
\label{gvsV}
\end{equation}
for every $\mu,\nu\in\{1,\ldots,m\}$.
\end{theorem}
Such a relation is well-known in literature (see for example \cite{Mardia,Rousson}); however, here we propose an alternative derivation (see Appendix \ref{appendixA}).

\bigskip

At this point, we proceed to show some properties of the volume defined in Definition \ref{defivol}. Given Remark \ref{rempar}, we would require the volume in \eqref{volume} to be invariant under labelling permutations of the modes. So, 
consider a point $\xi=(q_1,p_1,\ldots,q_N,p_N)\equiv (\xi_1,\ldots,\xi_N)\in \Gamma$ and a permutation $\sigma:\{1,\ldots,N\}\rightarrow\{1,\ldots,N\}$ such that $\xi_\sigma=(\xi_{\sigma(1)},\ldots,\xi_{\sigma(N)})$ is still a point in the phase space $\Gamma$. At the level of pdf in \eqref{PxT}, such a permutation acts by transforming the covariance matrix $V(\theta)$ in the following way 
\begin{equation}
V'(\theta)=\Pi^T\ V(\theta)\ \Pi,
\label{iso}
\end{equation}
where $V$ and $V^\prime$ are the covariance matrices of the state described by variables $\xi$ and $\xi_{\sigma(i)}$ respectively, and $\Pi$ is the permutation matrix given by $\Pi=({e}_{\sigma(1)},\ldots,{e}_{\sigma(N)})^T$, with ${e}_{j}$ denoting a row vector of length $2N$ with $1$ in the $j$th position and $0$ everywhere else.

Another feature we would require is the invariance of  the volume measure in \eqref{volume} under symplectic transformations $S\in\mbox{Sp}(2N,\RR)$, i.e. $S$ such that
\begin{equation}
\label{symplmatrix}
S^T\ \Omega\ S=\Omega,
\end{equation}
where $\Omega$ is the antisymmetric matrix defined in \eqref{J}. This requirement is motivated by the fact that from 
Eq. \eqref{symplmatrix} it follows that the uncertainty relation \eqref{QS} has an $\mbox{Sp}(2N,\RR)$ invariant form.

{The following proposition shows that both of these properties are satisfied by the volume $V$ defined in Definition \ref{defivol}}.

\begin{pro}\label{invariance}
If there exists a permutation matrix $\Pi$ (resp. a symplectic matrix $S$)
such that $V^\prime=\Pi^T\ V\ \Pi$ (resp. $V^\prime=S^T\ V\ S$),
then
\begin{equation}
{\cal V}(V^\prime)={\cal V}(V).
\end{equation}
\end{pro}
{In fact, the Fisher-Rao metric is invariant under more general transformations, namely the congruent transformations defined by elements of the general linear group
 (see for example \cite{Koranyi,Ohara,Shima}). However, for our purposes, it suffices to consider the behavior of the volume under the congruent transformations defined by permutation and symplectic matrices. See Appendix \ref{appendixB}.}


\section{Regularized volume}\label{sec4}

In general, the integration space $\Theta$ given in Def. \ref{defiphysical} is not bounded. However, this is not the only reason for the possible divergence of the integral \eqref{volume}. Indeed, let us recall that such an integral is computed by means of the volume element coming from the Fisher-Rao metric $G(\theta)$, which is
\begin{equation}
\nu_{ G}:=\sqrt{\det g(\theta)}\;d\theta^1\wedge\ldots\wedge d\theta^m.
\label{volumelement}
\end{equation}
Here, by Eq. \eqref{gvsV}, the entries of the $m\times m$ symmetric matrix $g(\theta)$  can be written in the form
\begin{equation}
\label{gvsDetV}
g_{\mu\nu}=F(V)\ \left(\det V\right)^{-2},
\end{equation} 
where $F(V)$ is what is left after grouping the common factor $\left(\det V\right)^{-2}$ in \eqref{gvsV}. Such a factor comes from the well-known relation $V^{-1}=\left(\det V\right)^{-1} \mbox{adj}(V)$, where $\mbox{adj}(V)$ denotes the adjunct of the matrix $V$.

Hence, we have  
\begin{equation}\label{detgvsV}
\det g(\theta)=\frac{1}{\left(\det V(\theta)\right)^{2 m}} \widetilde{F}(V(\theta)),
\end{equation}
where $\widetilde{F}(V(\theta))$ denotes a non-rational function of the coordinates $\theta^1,\ldots,\theta^m$.

Now, from Eq. \eqref{detgvsV} it is clear that the reasons for the possible divergence of the integral in \eqref{volume} are twofold: the set $\Theta$  in Def.\ref{defiphysical} is not compact because  the variables $\theta^l$ are unbounded from above, which makes the quantity $\widetilde{F}(V(\theta))$ divergent; furthermore, $\det g(\theta)$ diverges
since $\det V$ approaches zero for some $\theta^l\in \Theta$. 

It is then necessary to introduce a regularizing function $\varPhi(V)$ which eliminates these possible divergences. 
It should supply a kind of compactification of the parameter space and excludes the contributions of $\theta^l$ making $\det g(\theta)$ divergent. Such a function might stem on physical arguments related to finiteness of
energy.

In general, for an arbitrary Gaussian state $\rho$ with zero first moments, the trace of the covariance matrix
is directly linked to the mean energy per mode, namely ${\cal E}=\frac{1}{2N}\trace\left(V\right)$ \cite{Adesso}. Thereby, we propose to bound the parameter space with a suitable energy value of the state. To this end, we define a regularizing function as
\begin{equation}
\label{regenergy}
\varPhi(V):=H(\mathbf{E}-\trace(V))\ \log\left[1+ \left(\det V\right)^m\right],
\end{equation}
where $H(\cdot)$ denotes the Heaviside step function and $\mathbf{E}$ is a positive \textit{real} constant
(equal to $2N{\cal E}$).  

\medskip

With this regularizing function, we arrive at the following volume for sets of states:
\begin{definition}
\label{regenergyvoldefi}
Given a set of Gaussian states represented by a parameter space $\Theta$ as in Def. \ref{defiphysical}, we define its volume, regularized by the functional $\varPhi$,  to be
\begin{equation}
\label{volenergyreg}
\widetilde{\cal V}_{\varPhi}(V):=\int_{\Theta} \varPhi(V)\ \nu_G,
\end{equation}
where $\nu_G$ and $\varPhi(V)$ are given by Eqs. \eqref{volumelement},\eqref{regenergy} respectively.
\end{definition}  
The integral in \eqref{volenergyreg} is now meaningful. Indeed, 
we have the following results.
\begin{theorem}
\label{theorem:det-bound}
Let $E$ denote the constant $m \times m$ matrix defined by
\begin{align}
E_{\mu \nu} = \frac{1}{2}\tr[(\partial_{\mu}V)(\partial_{\nu} V)], \;\;\; 1 \leq \mu, \nu \leq m.
\end{align}
Let $\adj(V)$ denote the adjunct matrix of $V$. The Fisher-Rao information matrix $g$ satisfies
\begin{align}
\det g \leq \left(\frac{\lambda_{\max}[\adj(V)]}{\det V}\right)^{2m} \det(E)= \left(\frac{1}{\lambda_{\min}(V)}\right)^{2m}\det(E),
\end{align}
where $\lambda_{\max}[\adj(V)]$ denotes the largest eigenvalue of $\adj(V)$ and $\lambda_{\min}(V)$ denotes the smallest eigenvalue of $V$.
\end{theorem}
{\bf Proof.}\quad See Appendix \ref{appendixC}. $\hfill \Box$

\begin{corollary} 
\label{corollary:integral-boundenergy}
The regularized volume element satisfies
\begin{align}
\Phi(V)\sqrt{\det g} \leq \sqrt{\det E}\ H(\mathbf{E}-\tr(V)) \lambda^{m}_{\max}[\adj(V)]\frac{\log[1+(\det V)^m]}{(\det V)^m}.
\end{align}
Consequently, the integral
\begin{align}
\int_{\Theta}\Phi(V)\sqrt{\det g} d\theta,
\end{align}
is well-defined and bounded for any measurable subset $\Theta \subset \RR^m$ over which $V$ is positive definite.
\end{corollary}
{\bf Proof.}\quad See Appendix \ref{appendixC}. $\hfill \Box$

\begin{remark}
\label{remenergy}
Setting the energy of the Gaussian states to be smaller than or equal or to $\mathbf{E}$  results in an upper bound for parameter space $\Theta$. Furthermore, the singularity occurring when $\det V$ goes to zero is eliminated by the logarithm $\log\left[1+ \left(\det V\right)^m\right]$, using the well-known relation $\lim_{x\rightarrow 0}\frac{\log(1+x)}{x}=1$. So the integral \eqref{volenergyreg} is now meaningful.

However, the regularizing function $\varPhi(V)$ of Eq.\eqref{regenergy} is not invariant under symplectic transformations. Indeed, consider $S\in\mbox{Sp}(2N,\RR)$ and $V^\prime=S^T\ V\ S$; then, $\trace(V^\prime)=\trace(S\ S^T\ V)$, which is in general not equal to $\trace(V)$. 
\end{remark}

From Remark \ref{remenergy} the following issue arises: the well-known Williamson's Theorem \cite{Will} states that given a positive definite and symmetric $2N\times 2N$ matrix $V$, there exist a symplectic $2N\times 2N$ matrix $S$ and a diagonal and positive defined $2N\times 2N$ matrix $D$ such that $V=S^T DS$ \footnote{ Even though the {diagonalizing} symplectic matrix $S$ is not unique in general, if $S'$ is another diagonalizing matrix then there exists $U\in\mbox{U}(2N,\RR)$ such that $S'=SU$ \cite{deGosson07}.}. This implies that two Gaussian states are similar under the congruence transformation via a symplectic matrix. The regularizing function $\varPhi(V)$ is not invariant under such transformations.

In order to overcome this problem, we propose a different regularizing function $\Upsilon(V)$
which is devised by exploiting the relation given by Theorem \ref{main} and the property stated in  Proposition \ref{invariance}.
We define it as follows
\begin{equation}
\label{reg}
\Upsilon(V):=e^{-{\frac{1}{\kappa}}\tr \left[\adj(V)\right]}\ \log\left[1+ \left(\det V\right)^m\right],
\end{equation}
where  $\kappa$ is a \textit{real} positive number. Here, $\adj(V)$ denotes the adjunct of $V$, given by $\adj(V)=\det(V) V^{-1}$.

As required, the regularizing function $\Upsilon(V)$ fulfills the following properties:
\begin{pro}
\label{regpro}
Let $V$, $V^\prime$ be two covariance matrices and $\Pi$ be a permutation matrix (resp., $S$ be a symplectic matrix) such that
$V^\prime =\Pi^T\ V\ \Pi$
(resp. $V^\prime =S^T\ V\ S$), then
\begin{equation}
\label{regiso}
\Upsilon(V^\prime)=\Upsilon(V).
\end{equation}
\end{pro}
{\bf Proof.}\quad See Appendix \ref{appendixD}. $\hfill \Box$

\bigskip

With the regularizing function $\Upsilon(V)$, we arrive at the following volume for sets of states:
\begin{definition}
\label{regvoldefi}
Given a set of Gaussian states represented by a parameter space $\Theta$ as in Def. \ref{defiphysical}, we define its volume, regularized by the functional $\Upsilon$, to be
\begin{equation}
\label{volreg}
\widetilde{\cal V}_{\Upsilon}(V):=\int_{\Theta}\ \Upsilon(V)\ \nu_G,
\end{equation}
where $\nu_G$ and $\Upsilon(V)$ are given by Eqs. \eqref{volumelement},\eqref{reg} respectively.
\end{definition}

 The integral \eqref{volreg} is now meaningful. Indeed, as consequence of Theorem \ref{theorem:det-bound} we have the following result.
\begin{corollary} 
\label{corollary:integral-bound}
The regularized volume element satisfies
\begin{align}
\Upsilon(V)\sqrt{\det g} \leq \sqrt{\det E}\exp(-\trace[\adj(V)]) \lambda^{m}_{\max}[\adj(V)]\frac{\log[1+(\det V)^m]}{(\det V)^m}.
\end{align}
Consequently, the integral
\begin{align}
\int_{\Theta}\Upsilon(V)\sqrt{\det g} d\theta
\end{align}
is well-defined and bounded for any measurable subset $\Theta \subset \RR^m$ over which $V$ is positive definite.
\end{corollary}
{\bf Proof.}\quad See Appendix \ref{appendixD}. $\hfill \Box$ 
\begin{remark}\label{reminvariance}
Recalling that $\adj(V)=\det V\ V^{-1}$, we now have all the possible divergences in \eqref{volreg} suppressed: if $\lambda_{max}$ goes to infinity then the integrand is killed to zero by the exponential $e^{-\tr\left(\det (V)\ V^{-1}\right)}$; while if $\det V$ goes to zero then the singularity is eliminated by the logarithm $\log\left[1+ \left(\det V\right)^m\right]$, using the well-known relation $\lim_{x\rightarrow 0}\frac{\log(1+x)}{x}=1$.

Furthermore, thanks to Propositions \ref{main}, \ref{regpro} it follows that the regularized volume in \eqref{volreg} is invariant under permutation transformations and, if the states are quantum, is also invariant under symplectic transformations.
\end{remark}


\section{Example of bipartite states in two-mode system}\label{sec5}

We now apply the method proposed in the previous Section to a two-mode physical system ($N=2$). Hence, the elements of the statistical model ${\cal S}$ are the pdf$s$  in \eqref{statmodel},
where $V(\theta)$ is a $4\times 4$ covariance matrix and $\xi=(q_1,p_1,q_2,p_2)^T\in\Gamma=\RR^4$.
This implies that ${\cal S}=\{P(\xi;\theta)\}$ can be at most a $10$-dimensional statistical model.

According to Remark \ref{remprobdisj}, we consider that position and momentum variables of the same mode are not correlated. Thus the highest possible  dimension of ${\cal S}$ reduces from $m=10$ to $m=8$. Then the parameter space $\Theta$ is a subset of the linear space $\RR ^8$. 
The Gaussian classical and quantum states are represented by 
\begin{eqnarray}
\Theta_{\mbox{\tiny classic}}&=&\{\theta\in\RR^8\arrowvert V(\theta)>0\}\label{classic}\\
\Theta_{\mbox{\tiny quantum}}&=&\{\theta\in\RR^8\arrowvert V(\theta)+i \Omega\geq 0\}\label{quantum},
\end{eqnarray}
where $\Omega$ is the canonical symplectic \textit{real} $4\times 4$ matrix defined in \eqref{J}. 

The separable states have to respect not just the uncertainty relation $V(\theta)+i \Omega\geq 0$, but also the restriction $\widetilde{V}(\theta)+i\Omega\geq 0$, as stated in Theorem \ref{sepcor}, or equivalently
$V(\theta)+i\widetilde{\Omega}\geq 0$,
where $\widetilde{\Omega}=\Lambda_B \Omega \Lambda_B$ and $\Lambda_B$ is the partial transposition defined in \eqref{PPT2}.
Hence, the parameter space $\Theta$ representing Gaussian separable states is given by
\begin{equation}\label{separablebi}
\Theta_{\mbox{\tiny separable}}=\{\theta \in\RR^8\arrowvert V(\theta)+i \Omega\geq 0, V(\theta)+i\widetilde{\Omega}\geq 0 \}.
\end{equation}

As a consequence of Eqs. \eqref{quantum}, \eqref{separablebi} we also have that the Gaussian entangled states are represented by
\begin{equation}
\label{entangled}
\Theta_{\mbox{\tiny entangled}}=\Theta_{\mbox{\tiny quantum}}-\Theta_{\mbox{\tiny separable}}.
\end{equation}

Finally, to each set of states, classical, quantum, separable and entangled, we associate a Riemannian manifold given by ${\cal M}=\left(\Theta,G(\theta)\right)$, where $\Theta$ is specified by Eqs. \eqref{classic}, \eqref{quantum}, \eqref{separablebi}, \eqref{entangled}, respectively. Furthermore, $G(\theta)$ is the Fisher-Rao metric whose components $g_{\mu\nu}$ are given by \eqref{gvsV}, which is
an explicit expression in terms of the covariance matrix V.

However, the most general parametrization of a two-mode covariance matrix $V(\theta)$ is realized through its \textit{canonical form} and it only employs four parameters \cite{Simon},
\begin{equation}
\label{4par}
V(\theta)=\left(\begin{array}{cccc}
a&0&c&0\\
0&a&0&d\\
c&0&b&0\\
0&d&0&b
\end{array}\right)
\end{equation}
where, according to our notation, the only non zero parameters are $\theta^1=\theta^5=a\in\RR$, $\theta^8=\theta^{10}=b\in\RR$, $\theta^3=c\in\RR$ and $\theta^7=d\in\RR$. In this case, the domains of integration given by Eqs. \eqref{classic}, \eqref{quantum}, \eqref{separablebi}, apart from null sets, assume the following form:
\begin{eqnarray}
\label{Gaussian4parc}
\Theta_{\mbox{\tiny classic}}&=&\{(a,b,c,d)\in\RR^4\arrowvert\ a>0, b>0,\ -\sqrt{ab}<c<\sqrt{ab},\ -\sqrt{ab}<d<\sqrt{ab}\}\\\nonumber
\\
\label{Gaussian4parq}\Theta_{\mbox{\tiny quantum}}&=&\{(a,b,c,d)\in\RR^4\arrowvert\ a>1,\ 1<b<a,\ -\sqrt{c_1}<c<\sqrt{c_1},\ d_1\leq d\leq d_2\}\\\nonumber
&&\cup\ \{(a,b,c,d)\in\RR^4\arrowvert\ a>1,\ 1<a<b,\ -\sqrt{c_2}<c<\sqrt{c_2},\ d_1\leq d\leq d_2\} \\\nonumber
\\
\label{Gaussian4pars}\Theta_{\mbox{\tiny separable}}&=&\{(a,b,c,d)\in\RR^4\arrowvert\ a>1,\ b>1,\ -\sqrt{c_3}<c<0,\ d_1\leq d\leq -d_1 \}\\ \nonumber
&& \cup \ \{(a,b,c,d)\in\RR^4\arrowvert\ a>1,\ b>1,\ 0<c<\sqrt{c_3},\ -d_2\leq d\leq d_2 \},
\end{eqnarray}
with
\begin{eqnarray*}
&& c_1:=\frac{a}{b}\left(b^2-1\right),\ c_2:=\frac{b}{a}\left(a^2-1\right),\ c_3:=\frac{1-a^2-b^2+a^2b^2}{ab},\\
\\
&& d_1:=\frac{-c-\sqrt{\Delta}}{ab-c^2},\ d_2:=\frac{-c+\sqrt{\Delta}}{ab-c^2},\ \Delta:= c^2-(ab-c^2)\left[a b c^2 - (a^2 - 1)(b^2 - 1)\right].
\end{eqnarray*} 
Then, from the relation \eqref{gvsV} we can compute the Fisher-Rao metric 
$g(\theta)=g_{\mu\nu}d\theta^\mu\otimes d\theta^\nu$ ($\mu,\nu=1,\ldots,4$), whose components $g_{\mu\nu}$ 
result
\begin{eqnarray*}
&&g_{11} =\frac{b^2 (2a^2b^2 + c^4 + d^4 - 2 ab(c^2 + d^2)) }{2(ab-c^2)^2 (ab-d^2)^2},\quad
g_{12}=\frac{(a^2b^2 + c^2d^2)(c^2 + d^2)-4 a b c^2d^2}{2(ab-c^2)^2 (ab-d^2)^2},\\
\\
&& g_{13}=-\frac{bc}{(ab-c^2)^2}, \quad
g_{14}=-\frac{bd}{(ab-d^2)^2},\quad 
g_{22}=\frac{a^2 (2a^2b^2 + c^4 + d^4 - 2 ab(c^2 + d^2)) }{2(ab-c^2)^2 (ab-d^2)^2},\\
\\
&& g_{23}=-\frac{ac}{(ab-c^2)^2},\quad 
g_{24}=-\frac{ad}{(ab-d^2)^2},\quad 
g_{33}=\frac{ab+c^2}{(ab-c^2)^2},\quad 
g_{34}=0,\quad 
g_{44}=\frac{ab+d^2}{(ab-d^2)^2}.
\end{eqnarray*}

At this point we are able to compute the volumes of Gaussian states whether they are classic, quantum or entangled. First, it is evident from \eqref{Gaussian4parc},\eqref{Gaussian4parq} and \eqref{Gaussian4pars} that $\Theta_{\mbox{\tiny separable}}\subseteq\Theta_{\mbox{\tiny quantum}}\subseteq\Theta_{\mbox{\tiny classic}}$. 

Consider now the measure $\widetilde{\cal V}_{\varPhi}$ of Eq. \eqref{volenergyreg}; then, from Corollary \ref{corollary:integral-boundenergy} it follows that it is absolutely continuous with respect to the Lebesgue measure on $\RR ^4$. Also, it is monotonic. For all these reasons we have
\begin{align}\label{energyinc4}
\int_{\Theta_{\mbox{\tiny separable}}} \varPhi(V)\ \nu_G\leq \int_{\Theta{\mbox{\tiny quantum}}} \varPhi(V)\ \nu_G\leq \int_{\Theta{\mbox{\tiny classic}}} \varPhi(V)\ \nu_G,
\end{align}
for every $\mathbf{E}\in\RR^+$. Here, $\varPhi(V)=H(\mathbf{E}-\tr(V) )\ \log\left[1+(\det V)^4\right]=H(\mathbf{E}-2(a+b))\ \log\left[1+\left((ab-c^2)(ab-d^2)\right)^4\right].$

Given the sets of states as in \eqref{Gaussian4parc},\eqref{Gaussian4parq} and \eqref{Gaussian4pars}, after having obtained the set of entangled states by Eq.\eqref{entangled}, we can also compute the volume of the latter set by the measure $\widetilde{\cal V}_{\varPhi}$ of Eq.\eqref{volenergyreg}.  In Fig. \ref{RatioE} are reported the ratios of quantum over classical volumes, separable over classical, and entangled over classical. Figure \ref{RatioE} clearly shows the chain of inclusion 
$\widetilde{\cal V}_{\varPhi,{\rm entangled}}\subset \widetilde{\cal V}_{\varPhi,{\rm  separable}}\subset \widetilde{\cal V}_{\varPhi,{\rm quantum}}\subset \widetilde{\cal V}_{\varPhi,{\rm classical}}$ holding true for any value of $\mathbf{E}\in \RR^+$.  
{In addition all the volumes go to infinity when $\mathbf{E}\to\infty$, however this takes place at higher rate for ${\cal V}_{\varPhi,{\rm classical}}$, so that all the ratios approach zero when $\mathbf{E}\rightarrow+\infty$.
Analogously, for $\mathbf{E}\to 0$ all sets becomes empty and the ratios become zero.}

\begin{figure}[h!] \centering
\includegraphics[width=7cm,height=6cm,scale=1.1]{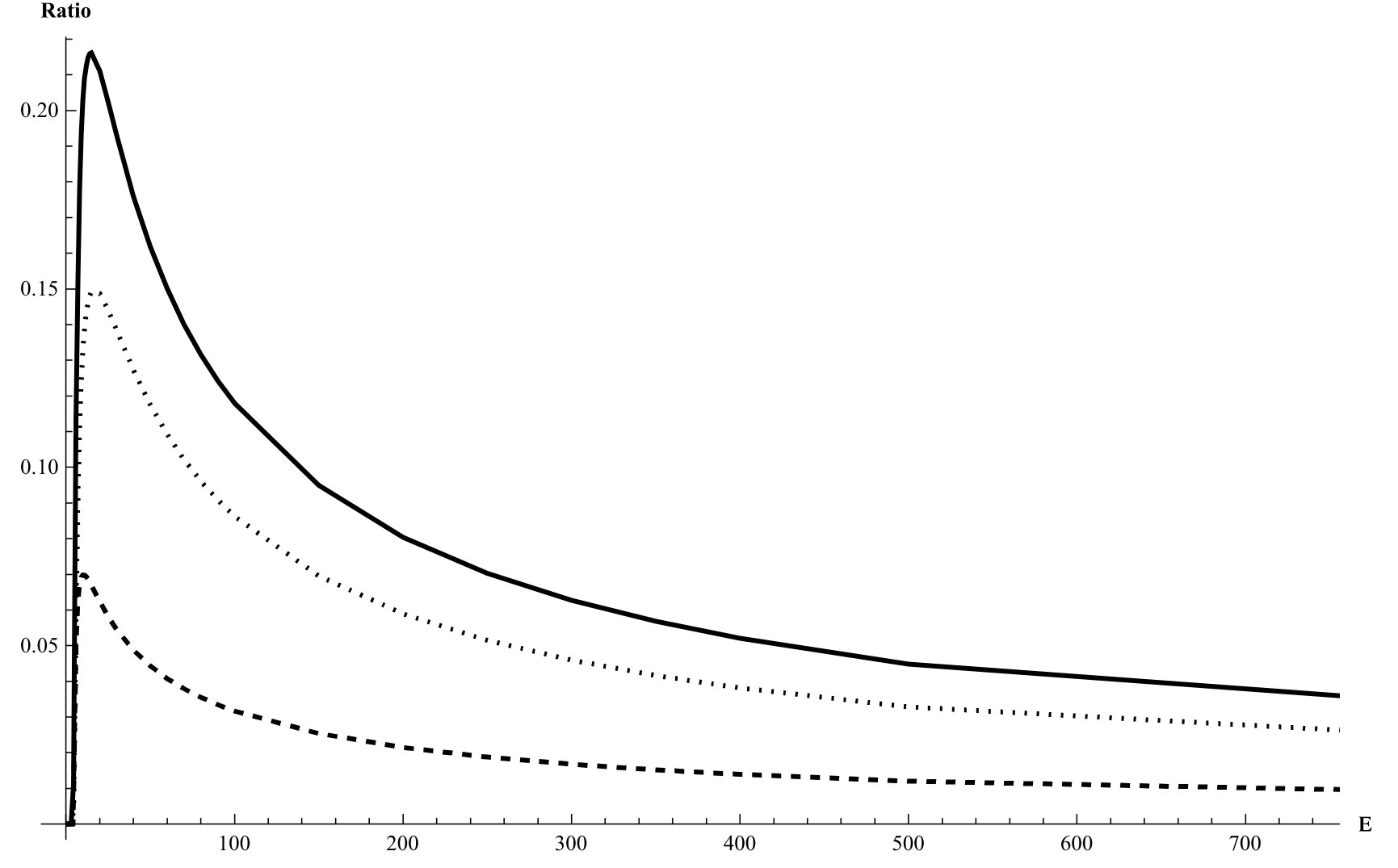}
\caption{Ratios of volumes $\widetilde{\cal V}_{\varPhi}$ of Eq.\eqref{volenergyreg} vs $\mathbf{E}$ (from top to bottom  quantum over classical, separable over classical, and entangled over classical).}
\label{RatioE}
\end{figure}

\begin{remark}
The form of the covariance matrix $V$ in \eqref{4par} is a very special one. Indeed, any {$4\times 4$} covariance matrix can be brought to that by a suitable transformation corresponding to some element of $\mbox{Sp}(2,\RR)\times \mbox{Sp}(2,\RR)$. 
\end{remark}

Unfortunately, the function $\varPhi$ just employed to get the inclusions \eqref{energyinc4} is not invariant under $\mbox{Sp}(2,\RR)\times \mbox{Sp}(2,\RR)$ transformations. 

On the other hand, the parameter space given in Eq. \eqref{separablebi}, which describes the Gaussian separable states of a two-mode system, has a $\mbox{Sp}(2,\RR)\otimes \mbox{Sp}(2,\RR)\subset \mbox{Sp}(4,\RR)$ invariant form. Hence, the function $\varPhi$ seems to be unsuitable to describe the volume of Gaussian states. In contrast, the regularizing function $\Upsilon $ of Eq.\eqref{reg} is invariant under local symplectic transformations, as it immediately follows from Prop. \ref{regpro}. For this reason, we propose the volume \eqref{volreg} as a suitable measure assessing differences among the Gaussian states. Again, we obtain strict inclusions among the volumes of classical, quantum, entangled and separable states, namely 
\begin{align}\label{energyinc}
\int_{\Theta_{\mbox{\tiny separable}}} \Upsilon(V)\ \nu_G\leq \int_{\Theta{\mbox{\tiny quantum}}}  \Upsilon(V)\ \nu_G\leq \int_{\Theta{\mbox{\tiny classic}}}  \Upsilon(V)\ \nu_G,
\end{align}
with $ \Upsilon(V)=e^{-\frac{1}{\kappa}\left(2 a^2 b + a (2 b^2 - c^2 - d^2) - b (c^2 + d^2)\right)}\ \log\left[1+\left((ab-c^2)(ab-d^2)\right)^4\right]$ and for all $\kappa\in\mathbb{R}_+$.
It is worth stressing that due to the regularization, the important quantities are not the absolute values of volumes, but rather their ratios.

In Fig. \ref{Ratio} are reported the ratios of the volumes of different sets of Gaussian states computed through the measure \eqref{volreg}. In particular the ratios of Quantum over Classic, Entangled over Classic and Separable over Classic are shown vs $\kappa$.

\begin{figure}[h!] \centering
\includegraphics[width=7cm,height=6cm,scale=1.1]{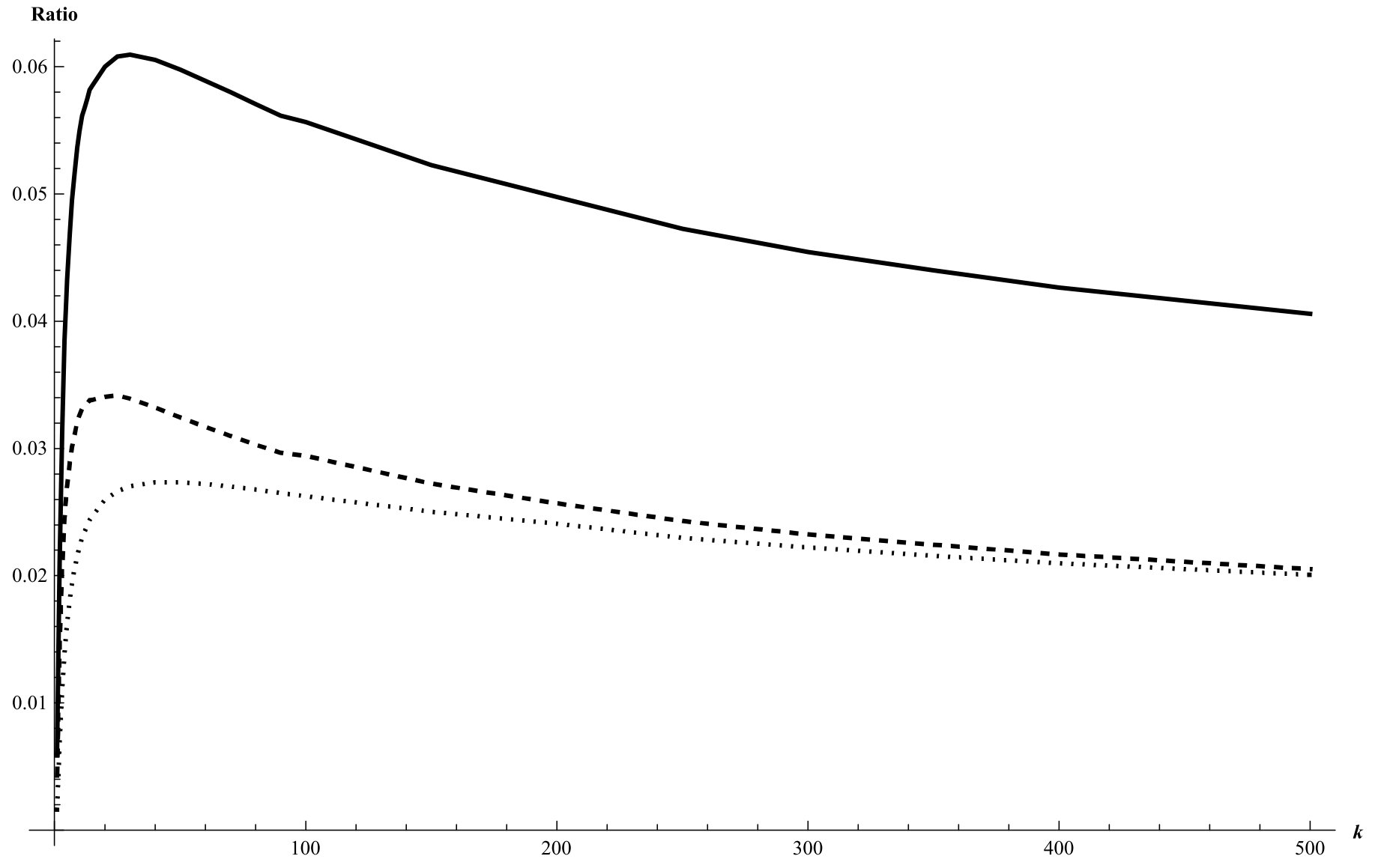}
\caption{Ratios of volumes $\widetilde{\cal V}_{\Upsilon}$ of Eq.\eqref{volreg} vs $\kappa$ (from top to bottom  quantum over classical, entangled over classical, and separable over classical.}
\label{Ratio}
\end{figure}

Figure \ref{Ratio} clearly shows the chain of inclusion 
$\widetilde{\cal V}_{\Upsilon,{\rm separable}}\subset \widetilde{\cal V}_{\Upsilon,{\rm entangled}}\subset \widetilde{\cal V}_{\Upsilon,{\rm quantum}}\subset \widetilde{\cal V}_{\Upsilon,{\rm classical}}$ holding true for any value of $\kappa$. For $\kappa\to 0$ all sets becomes empty and the ratios become zero, while for $\kappa\to\infty$ they tend to asymptotic values.

Both Figs. \ref{RatioE} and \ref{Ratio} put forward a non-monotonic behavior of the volume ratios. This effect turns out to be of purely geometric nature and has to be ascribed to the curved metric \eqref{gFR} (by contrast one can check that with standard Euclidean metric the behavior is monotonic).
By comparing Fig. \ref{RatioE} and Fig. \ref{Ratio} it is worth noticing the different hierarchies of volumes appearing there. This fact can be explained as follows. First we recall that a symplectic transformation $S$ acts by congruence on a covariance matrix $V\rightarrow SVS^T$, hence the set of Gaussian states can be thought as the orbit for the action of ${\rm Sp}(2N,\mathbb{R})$ on a seed $V$.
Then, the regularization \eqref{reg} provides a ``homogeneous'' cut-off on the space of Gaussian states and as consequence a hierarchy of volumes similar to that of finite (low) dimensional case is obtained (in Ref. \cite{Zycz} it was shown that the volume of separable states is contained in the volume of entangles states, which in turn is contained in the volume of quantum states).
In contrast, the regularization \eqref{regenergy} provides a non ``homogeneous'' cut-off on the space of Gaussian states. In fact it only cuts the states arising from the action of elements of the non compact subgroup of $\mbox{Sp}(2N,\RR)$, leaving unaffected states arising from the action of the compact subgroup of $\mbox{Sp}(2N,\RR)$ (actually $\mbox{Sp}(2N,\RR)\cap \mbox{SO}(2N,\RR) $). Consequently the hierarchy of volumes obtained differs from that of the finite (low) dimensional case.


\section{Conclusion}\label{sec6}

\noindent In the present work, we tackled the problem of evaluating the volume of Gaussian physical states, both classical and quantum. The relevance of considering Gaussian states is twofold: first, Gaussian states are the most commonly experimentally used CV states. Second, Gaussian quantum states are represented in the phase space picture of quantum mechanics as proper pdf$s$. Hence, Gaussian classical states are pdf$s$ in phase space and Gaussian quantum states are pdf$s$ coming from Wigner functions in phase space. Thereby, dealing with pdf$s$, Information Geometry appears as a natural and unifying approach for evaluating volume of classical and quantum states. 

By exploiting methods of Information Geometry, we associated manifolds to different sets of states; thus there is one manifold corresponding to classical states, one manifold to quantum states, another one to separable and another to entangled states. The key point in setting up such connections was that each set of states can be obtained by considering the pdf$s$ parametrized by the entries of the covariance matrix. Then the manifolds are exactly the parameter spaces obtained by imposing constraints on the covariance matrix in order to describe classical or quantum states. Concerning separable and entangled states, the question is more delicate. Indeed, there is no general criterion to characterize multipartite entangled states. Nonetheless, by reducing to bipartite systems one could use the condition \eqref{separablestates} (which turns into an operational condition for two-mode systems thanks to \eqref{separable}). Then, in this case we can also associate a manifold to separable states and a manifold to entangled states.
Next we endowed each of these manifolds with a Riemannian metric, the Fisher-Rao one. Thus it was natural to define the volume of a set of Gaussian states as the volume \eqref{volume} of the Riemannian manifold associated to it. 

Thanks to Eq. \eqref{gvsV}, we were able to show that the volume measure introduced in \eqref{volume} is invariant under labelling permutations of modes; moreover, we proved that it has a symplectic invariant form. These results showed that the volume measure in \eqref{volume} is suitable for estimating the volume of Gaussian states. 
However, since we analysed infinite dimensional systems, problems arose from the non-compactness of the support of the states. We overcome this difficulty first of all by resorting to an energy constraint. Hence we defined the regularizing function \eqref{regenergy}. However, such kind of regularization turns out to be not invariant under symplectic congruence. Then we introduced a different regularizing function, namely the one in \eqref{reg}, which came about by figuring out the functional relation given by Eq. \eqref{gvsDetV}. We proved that also such a function has permutation and symplectic invariant form (Prop. \ref{regpro}).

Accordingly with these regularizing functions we have explicitly evaluated the volume of two-mode Gaussian states. 
Note that it is not the values of the volumes per s\'e that are really relevant, but rather the ratios between the volumes of the various sets. As such we presented in Figs. \ref{RatioE} and \ref{Ratio} the ratios with respect to the volume of classical states. The Figures show different hierarchies of volumes with the one in Fig. \ref{Ratio} resembling  that of finite dimensional systems (at least for quantum states, see e.g. \cite{Zycz}).
In both cases the ratios depend on the cut-off parameter in a non-monotonic way due to geometric effects. This makes evident a rich structure for sets of Gaussian states.

Finally, the presented volume measure could also be applied to three-mode systems, for which an operational criterion to distinguish separable states among all the quantum states is well-known \cite{Giedke}. Indeed, in such way the parameter space in Def. \ref{defiphysical} can be implemented for separable states as well as for classical and quantum states. Thereby, the volumes can be computed. Beyond that,  a necessary and sufficient criterion to describe Gaussian separable states still lacks, hence the introduced volume measure can only be useful to provide bounds on the volume of sets of multipartite states.

\acknowledgments

\noindent Part of this work was done while DF was visiting the Department of Physics of the Czech Technical University in Prague. DF warmly thanks Igor Jex and his collaborators for useful discussions as well as for kind hospitality. DF and SM also acknowledge fruitful discussions with Federico Maiolini.


\appendix

\section{Proof of Theorem \ref{main}}\label{appendixA}

\noindent Let us notice that because of the form of $P(\xi;\theta)$ in \eqref{statmodel}, the expression in Eq. \eqref{gFR} involves a Gaussian integral. However, before evaluating it, let us study the function
\begin{equation}
\label{fmunu}
f_{\mu\nu}(\xi):=\partial_\mu\log  P(\xi;\theta)\partial_\nu\log P(\xi;\theta).
\end{equation}
By means of logarithm's properties we can write
\begin{equation}
\log[P(\xi;\theta)]=-\frac 1 2\Bigg[\log[(2\pi)^N\det V(\theta)]+\sum_{\alpha,\beta=1}^{2N} V_{\alpha\beta}^{-1}(\theta)\xi_\alpha \xi_\beta\Bigg],
\label{log}
\end{equation}
where $V_{\alpha\beta}^{-1}(\theta)$ is the entry $\alpha\beta$ of the inverse of the covariance matrix $V(\theta)$. Then the derivative $\partial_\mu$ of Eq. \eqref{log} reads
\begin{equation}
\label{partmu1}
\partial_\mu \log[P(\xi;\theta)]=-\frac 1 2\Bigg[\frac{\partial_\mu(\det V)}{\det V}+\sum_{\alpha,\beta=1}^{2N} \partial_\mu(V_{\alpha\beta}^{-1})\xi_\alpha \xi_\beta\Bigg]. 
\end{equation}
Recall that the following relation holds
\begin{equation}
\label{partmu2}
\partial_\mu(\det V(\theta))=\det V(\theta) \;\tr\left[V^{-1}(\theta)\,\partial_\mu(V(\theta))\right].
\end{equation}
Hence, using \eqref{partmu1} and \eqref{partmu2}, we arrive at
\begin{eqnarray}
f_{\mu\nu}(\xi)&=&
\frac{1}{4} \Bigg[\tr \left[V^{-1}(\theta)\,\partial_\mu(V(\theta))\right]+\sum_{\alpha,\beta=1}^{2N}\partial_\mu \left(V_{\alpha\beta}^{-1}(\theta)\right)\xi_\alpha\xi_\beta\Bigg]\nonumber\\
&&\times \Bigg[\tr \left[V^{-1}(\theta)\,\partial_\nu(V(\theta))\right]+\sum_{\alpha,\beta=1}^{2N}\partial_\nu \left(V_{\alpha\beta}^{-1}(\theta)\right)\xi_\alpha\xi_\beta\Bigg].
\label{f}
\end{eqnarray}

For {an analytic} function $f(\xi)$ and a symmetric definite-positive $2 N\times 2 N$ matrix
$A$ it results
\begin{equation}
\int d\xi f(\xi)e^{\Big[-\frac{1}{2}\sum_{i,j=1}^{2N}A_{ij}\xi_i\xi_j\Big]}=\sqrt{\frac{(2\pi)^{2 N}}{\det A}}\exp\left[\frac 1 2 \sum_{i,j=1}^{2 N} 
A_{ij}^{-1}\frac{\partial}{\partial \xi_i}\frac{\partial}{\partial \xi_j}\right]f |_{\xi=0},
\label{IG}
\end{equation}
where $A_{ij}^{-1}$ is the entry $ij$ of the inverse of the matrix $A$ and the exponential means the power series over its argument (the differential operator). {Indeed, by expanding $f(\xi)$ we have that $f(\xi)=\sum_{\alpha}\frac{{\cal D}^\alpha f(0)}{\alpha!}\xi^\alpha=e^{\Sigma\xi}$, where $\Sigma=({\cal D} f)(0)$, with $\alpha$ a multi-index and ${\cal D}$ the differential operator with respect local coordinates $\xi_1,\ldots\xi_N$. At this point the 
left hand side of Eq.\eqref{IG} can be written as
\begin{align*}
\int d\xi e^{\Big[-\frac{1}{2}\sum_{i,j=1}^{2N}A_{ij}\xi_i\xi_j\Big]+\Sigma\xi}.
\end{align*}
Then by performing an orthogonal transformation ${\cal O}$, we have 
\begin{align*}
-\frac{1}{2} \xi^T A\xi+\Phi\xi\rightarrow -\frac{1}{2} y^T D y+\Sigma{\cal O}y,
\end{align*}
where $D$ is diagonal matrix with elements eigenvalues of $A$. Finally, through some algebras, we arrive at
\begin{align*}
\int d\xi e^{\Big[-\frac{1}{2}\sum_{i,j=1}^{2N}A_{ij}\xi_i\xi_j\Big]+\Sigma\xi}=\sqrt{\frac{(2\pi)^{2N}}{\det A}}\ e^{-\frac{1}{2}\Sigma^T A^{-1}\Sigma},
\end{align*}
which gives exactly Eq.\eqref{IG}. }

Inserting the expression of $P(\xi;\theta)$ in \eqref{statmodel} into the relation \eqref{gFR} and employing Eq.\eqref{IG} we find
\begin{equation}
\frac{1}{\sqrt{(2\pi)^{2N}\det V}}\int d\xi f_{\mu\nu}(\xi) e^{\left[-\frac 1 2 \xi^T V^{-1} \xi\right]}
=\exp\left[\frac 1 2 \sum_{i,j=1}^n 
V_{ij}\frac{\partial}{\partial \xi_i}\frac{\partial}{\partial \xi_j}\right]f_{\mu\nu} |_{\xi=0}.
\label{Gint}
\end{equation}

We are now going to evaluate the Gaussian integrals in Eq.\eqref{gFR} by means of the following Lemma.
\begin{lemma}\label{lemma1}
Let be $D:=\frac 1 2 \sum_{i,j=1}^{2N}V_{ij}\frac{\partial}{\partial \xi_i}\frac{\partial}{\partial \xi_j}$, expanding the right-hand side of \eqref{Gint} we have 
\begin{equation}
g_{\mu\nu}(\theta)=f_{\mu\nu}(0)+Df_{\mu\nu} |_{\xi=0}+\frac 1 2 D^2f_{\mu\nu} |_{\xi=0},
\label{gexp}
\end{equation}
with
\begin{eqnarray}
Df_{\mu\nu}|_{\xi=0}&=&\frac{1}{4}\tr \left[V^{-1}(\theta)\,\partial_\mu(V(\theta))\right]\tr \left[V(\theta)\,\partial_\nu(V^{-1}(\theta))\right]\nonumber\\
&&+\frac{1}{4}\tr \left[V(\theta)\,\partial_\mu(V^{-1}(\theta))\right]\tr \left[V^{-1}(\theta)\,\partial_\nu(V(\theta))\right],
\label{D}
\end{eqnarray}
and
\begin{eqnarray}
\frac 1 2 D^2f_{\mu\nu}|_{\xi=0}&=&\frac{1}{4}\tr \left[V(\theta)\,\partial_\mu(V^{-1}(\theta))\right]\tr \left[V(\theta)\,\partial_\nu(V^{-1}(\theta))\right]\nonumber\\
&&+\frac{1}{2}\tr \left[V(\theta)\partial_\mu(V^{-1}(\theta))V(\theta)\partial_\nu(V^{-1}(\theta))\right].
\label{D2}
\end{eqnarray}
\end{lemma}

{\bf Proof.}\quad From Eq. \eqref{f}, with $i,j\in\{1,\ldots,2N\}$, by a straightforward calculation we have
\begin{eqnarray*}
\frac{\partial}{\partial \xi_i}\Bigg(\frac{\partial f_{\mu\nu}}{\partial \xi_j}\Bigg)(\xi)&=&\frac{1}{2}\partial_\mu V_{ij}^{-1}\left(\tr \left[V^{-1}\,\partial_\nu V\right]+\sum_{\alpha,\beta=1}^{2N}\partial_\nu \left(V_{\alpha\beta}^{-1}\right)\xi_\alpha\xi_\beta\right)\\
&&+ \left(\sum_{\beta=1}^{2N}\partial_\mu\left(V_{i\beta}^{-1}\right)\xi_\beta\right)\left(\sum_{\beta=1}^{2N}\partial_\nu\left(V_{j\beta}^{-1}\right)\xi_\beta\right)\\
&&+\left(\sum_{\beta=1}^{2N}\partial_\mu\left(V_{j\beta}^{-1}\right)\xi_\beta\right)\left(\sum_{\beta=1}^{2N}\partial_\nu\left(V_{i\beta}^{-1}\right)\xi_\beta\right)\\
&&+\frac{1}{2}\partial_\nu V_{ij}^{-1}\left(\tr \left[V^{-1}\,\partial_\mu V\right]+\sum_{\alpha,\beta=1}^{2N}\partial_\mu \left(V_{\alpha\beta}^{-1}\right)\xi_\alpha\xi_\beta\right).
\end{eqnarray*}
Taking the sum over $i,j\in\{1,\ldots,2N\}$ and evaluating the above expression at $\xi=0$, we obtain
\begin{eqnarray*}
Df_{\mu\nu}|_{\xi=0}&=&\frac{1}{4}\sum_{i,j=1}^{2N}V_{ij}\partial_\mu V_{ij}^{-1}\tr \left[V^{-1}\,\partial_\nu V\right]\\
&&+\frac{1}{4}\sum_{i,j=1}^{2N}V_{ij}\partial_\nu V_{ij}^{-1}\tr \left[V^{-1}\,\partial_\mu V\right].
\end{eqnarray*}
Now, recall that 
\begin{enumerate}
\item $\tr \left[AB\right]=\sum_{i,j=1}^{2N}A_{ij}B_{ij}$, for any pair of $N\times N$ matrices $A,B$;
\item $\partial_\mu(V(\theta))=\Big[\frac{\partial V_{ij}}{\partial\theta_\mu}\Big]_{ij}$ for any matrix $V$.
\end{enumerate}
Hence, we get Eq. \eqref{D}.

Furthermore, letting $i,j,h,k\in\{1,\ldots,2N\}$, we have 
\begin{eqnarray*}
\frac{\partial}{\partial \xi_h}\Bigg(\frac{\partial}{\partial \xi_k}\frac{\partial}{\partial \xi_i}\frac{\partial }{\partial \xi_j}f_{\mu\nu}\Bigg)(\xi)&=&\partial_\mu V_{ij}^{-1}\ \partial_\nu V_{hk}^{-1}+\partial_\mu V_{ih}^{-1}\ \partial_\nu V_{jk}^{-1}\nonumber\\
&&+\partial_\mu V_{ik}^{-1}\ \partial_\nu V_{jh}^{-1}+\partial_\mu V_{jh}^{-1}\ \partial_\nu V_{ik}^{-1}\nonumber\\
&&+\partial_\mu V_{jk}^{-1}\ \partial_\nu V_{ih}^{-1}+\partial_\mu V_{hk}^{-1}\ \partial_\nu V_{ij}^{-1}.
\label{derivata}
\end{eqnarray*}
Taking the sum over $i,j,h,k\in\{1,\ldots,2N\}$ we obtain
\begin{eqnarray*}
\frac 1 2 D^2f_{\mu\nu}(\xi)&=&\frac{1}{8} \Bigg\{\sum_{i,j,h,k}V_{ij}V_{hk}\ \partial_\mu V_{ij}^{-1}\ \partial_\nu V_{hk}^{-1}+\sum_{i,j,h,k}V_{ij}V_{hk}\ \partial_\mu V_{ih}^{-1}\ \partial_\nu V_{jk}^{-1}\\
&&+\sum_{i,j,h,k}V_{ij}V_{hk}\ \partial_\mu V_{ik}^{-1}\ \partial_\nu V_{jh}^{-1}+\sum_{i,j,h,k}V_{ij}V_{hk}\ \partial_\mu V_{jh}^{-1}\ \partial_\nu V_{ik}^{-1}\\
&&+\sum_{i,j,h,k}V_{ij}V_{hk}\ \partial_\mu V_{jk}^{-1}\ \partial_\nu V_{ih}^{-1}+\sum_{i,j,h,k}V_{ij}V_{hk}\ \partial_\mu V_{hk}^{-1}\ \partial_\nu V_{ij}^{-1} \Bigg\}\\
&=&\frac{1}{8} \Bigg\{2\tr \left[V\partial_\mu V^{-1}\right]\tr \left[V\partial_\nu V^{-1}\right]+4\tr \left[V\ \partial_\mu\ V^{-1}\ V\ \partial_\nu V^{-1}\right]\Bigg\}.
\end{eqnarray*}
Finally, thanks to the above expression of $\frac 1 2 D^2f_{\mu\nu}(\xi)$, we have that the expansion in the 
right-hand side of  Eq.\eqref{Gint} only contains terms up to the second order. $\hfill\Box$

\bigskip

At this point, collecting the results in Lemma \ref{lemma1} together Eq. \eqref{f} evaluated in $\xi=0$, we obtain
\begin{eqnarray}
g_{\mu\nu}&=&\frac 1 4 
\left[ \tr\left(V^{-1}\ \partial_\mu V\right)
+ \tr\left(V \partial_\mu V^{-1}\right) \right]
\left[ \tr\left(V^{-1}\ \partial_\nu V\right)
+ \tr\left(V \partial_\nu V^{-1}\right) \right]
\nonumber\\
&+&\frac{1}{2}\tr \left(V\ \partial_\mu\ V^{-1}\ V\ \partial_\nu V^{-1}\right).
\end{eqnarray}
Then, the statement of  Theorem \ref{main} easily follows from relation $\partial_\mu V^{-1}=-V^{-1}(\partial\mu V) V^{-1}$ .$\hfill\Box$


\section{Proof of Proposition \ref{invariance}}\label{appendixB}

\noindent Let us consider the permutation $\sigma:\{1,\ldots,N\}\rightarrow\{1,\ldots,N\}$ and the corresponding permutation matrix $\Pi$ which entails a labelling permutation in the phase space $\Gamma$, i.e $\Pi:\xi\in\Gamma\rightarrow\xi_\sigma\in\Gamma$. From the $P(\xi;\theta)$ in \eqref{statmodel},  it follows that such a permutation acts on the covariance matrix $V$ in the following manner:
\begin{equation*}
V\rightarrow \Pi\ V\ \Pi^T.
\end{equation*}
So, let $V^\prime(\theta)$ and $V(\theta)$ be two parametrized covariance matrices and $\Pi$ a permutation matrix such that $V^\prime(\theta) =\Pi\ V(\theta)\ \Pi^T$. Let $\Theta$ and $\Theta^\prime$ be the parameter spaces corresponding to $V(\theta)$ and $V^\prime(\theta)$, respectively. Then there exists a diffeomorphism $\varphi:\Theta\rightarrow\Theta^\prime$ with Jacobian $J_\varphi$ such that $\arrowvert\det J_\varphi\arrowvert=1$. Therefore we have
\begin{eqnarray*}
{\cal V}(V^\prime)=\int_{\Theta^\prime} d\theta\ \sqrt{\det g^\prime(\theta)}
=\int_{\Theta} d\theta\ \sqrt{\det g(\theta)}={\cal V}(V),
\end{eqnarray*}
where we used the equality $\det g^\prime(\theta)=\det g(\theta)$ intending $g^\prime(\theta)$ as the Fisher-Rao information matrix corresponding to $V^\prime(\theta)$.  

Actually, showing that $\det g^\prime(\theta)=\det g(\theta)$ we are proving a stronger relation between the Fisher-Rao metrics corresponding to $V$ and $V^\prime$. In fact the following relation holds true,
\begin{eqnarray*}
\mbox{tr}\left[V^\prime\ \partial_\mu\ (V^\prime)^{-1}\ V^\prime\ \partial_\nu (V^\prime)^{-1}\right]&=&\mbox{tr}\left[\Pi V\Pi^T\ \partial_\mu\ (\Pi V\Pi^T)^{-1}\ \Pi V\Pi^T\ \partial_\nu (\Pi V\Pi^T)^{-1}\right]\nonumber\\
&=&\mbox{tr}\left[V\ \partial_\mu\ V^{-1}\ V\ \partial_\nu V^{-1}\right],
\end{eqnarray*}
where we used the independence of $\Pi$  from $\theta$ and the invariance of the trace under cyclic permutation.

Then, recalling the relation \eqref{gvsV}, we arrive at $g_{\mu\nu}^\prime=g_{\mu\nu}$, for every $\mu,\nu\in\{1,\ldots,m\}$. Here, $g_{\mu\nu}^\prime$ denotes the component of the metric corresponding to $V^\prime$. Thereby, $\det g^\prime(\theta)=\det g(\theta)$ trivially holds true.

\bigskip

Focusing on the quantum states, it is well-known that the uncertainty relation $V+i\Omega\geq 0$ is invariant under symplectic transformation \cite{Simon}. So, let us consider two parametrized covariance matrices $V^\prime(\theta)$ and $V(\theta)$ and a symplectic matrix $S$ such that $V^\prime(\theta) =S\ V(\theta)\ S^T$. Then, the parameter spaces $\Theta_{\mbox{\tiny quantum}}^\prime$ and $\Theta_{\mbox{\tiny quantum}}$, corresponding to those different matrices, coincides.

Furthermore, we have 
\begin{eqnarray*}
\mbox{tr}\left[V^\prime\ \partial_\mu\ (V^\prime)^{-1}\ V^\prime\ \partial_\nu (V^\prime)^{-1}\right]&=&\mbox{tr}\left[SVS^T\ \partial_\mu\ (SVS^T)^{-1}\ SVS^T\ \partial_\nu (SVS^T)^{-1}\right]\nonumber\\
&=&\mbox{tr}\left[V\ \partial_\mu\ V^{-1}\ V\ \partial_\nu V^{-1}\right],
\end{eqnarray*}
where we used the independence of $S$ from $\theta$ and the invariance of the trace under cyclic permutation. Thus, from \eqref{gvsV} it trivially follows that,
\begin{eqnarray*}
{\cal V}(V^\prime)=\int_{\Theta_{\mbox{\tiny quantum}}^\prime} d\theta\ \sqrt{\det g^\prime(\theta)}
=\int_{\Theta_{\mbox{\tiny quantum}}} d\theta\ \sqrt{\det g(\theta)}
={\cal V}(V),
\end{eqnarray*}
where $g^\prime(\theta)$ and $g(\theta)$ denote the Fisher-Rao information matrix corresponding to the covariance matrices $V^\prime$ and $V$, respectively.
$\hfill\Box$


\section{Proof of Theorem \ref{theorem:det-bound} and Corollary \ref{corollary:integral-boundenergy}}\label{appendixC}

The proof of Theorem \ref{theorem:det-bound} relies on an inequality of 
the determinants of Gram matrices \cite{Yamada:2013}.
We recall that for an inner product space $\H$ and a set of points $\{x_1, \ldots, x_n\}$ in $\H$, the Gram matrix of this set is the $n \times n$ matrix defined by $G(x_1, \ldots, x_n) = (\la x_i, x_j\ra)_{i,j=1}^n$.
 
\begin{lemma}\cite{Yamada:2013}
\label{lemma:Gram-inequality}
Let $\H_1, \H_2$ be two inner product spaces and $T:\H_1 \mapto \H_2$ be a bounded linear operator. Let $\{x_1, \ldots, x_n\} \in \H_1$ be an arbitrary set. Then
\begin{align}
\det G(Tx_1, \ldots, Tx_n) \leq ||T||^{2n}\det G(x_1, \ldots, x_n),
\end{align}
where $||T||$ denote the operator norm of $T$.
\end{lemma}
\begin{proof}[\textbf{Proof of Theorem \ref{theorem:det-bound}}]
Since $V$ is symmetric, positive definite, its adjunct, given by $\adj(V) = \det(V)V^{-1}$ is also symmetric, positive definite.  
We have
\begin{align*}
g_{\mu \nu} &= \frac{1}{2}\trace[V^{-1}(\partial_{\mu}V)V^{-1}(\partial_{\nu}V)]
= \frac{1}{2}\frac{1}{(\det V)^2}\trace[\adj(V)(\partial_{\mu}V)\adj(V)(\partial_{\nu}V)]
\\
&= \frac{1}{(\det V)^2}\tilde{g}_{\mu \nu},
\end{align*}
where the matrix $(\tilde{g}_{\mu, \nu})_{\mu, \nu =1}^m$ is given by
\begin{align*}
\tilde{g}_{\mu \nu} &=\frac{1}{2}\trace[\adj(V)(\partial_{\mu}V)\adj(V)(\partial_{\nu}V)]
\\
&= \frac{1}{2}\trace[(\adj(V)^{1/2}(\partial_{\mu}V)\adj(V)^{1/2})(\adj(V)^{1/2}(\partial_{\nu}V)\adj(V)^{1/2})]
\\
& = \frac{1}{2}\la (\adj(V)^{1/2}(\partial_{\mu}V)\adj(V)^{1/2}), (\adj(V)^{1/2}(\partial_{\nu}V)\adj(V)^{1/2})\ra_{F},
\end{align*}
with $\la\cdot,\cdot\ra_F$ denoting the Frobenius inner product given by $\la A, B\ra_F =\trace\left[\left(B^T A\right)\right]$ for any matrices $A, B$ with same dimension.

Consider the linear operator $T_V:\R^{2N \times 2N} \mapto \R^{2N \times 2N}$, with $\R^{2N \times 2N}$ under the Frobenius inner product, defined by
\begin{align}
T_VA = \adj(V)^{1/2}A\adj(V)^{1/2}.
\end{align}
Then
\begin{align*}
||T_VA||_F  &= ||\adj(V)^{1/2}A\adj(V)^{1/2}||_F \leq ||\adj(V)^{1/2}||\;||A\adj(V)^{1/2}||_{F} 
\\
& \leq ||\adj(V)||\;||A||_F,
\end{align*}
with equality if $A = I$, where we have used the property that $||\adj(V)|| = ||\adj(V)^{1/2}||^2$ by the symmetric, positive definiteness of $\adj(V)$.  Thus 
\begin{align}
||T_V|| = ||\adj(V)||. 
\end{align}
Then we have
\begin{align*}
\tilde{g}_{\mu \nu} = \frac{1}{2}\la T(\partial_{\mu}V), T(\partial_{\nu}V)\ra_F.
\end{align*}
Let $E$ be the $m \times m$ matrix defined by
\begin{align*}
E_{\mu \nu} = \frac{1}{2}\trace[(\partial_{\mu}V)(\partial_{\nu} V)] = \frac{1}{2}\la (\partial_{\mu}V), (\partial_{\nu} V)\ra_F.
\end{align*}
By Lemma \ref{lemma:Gram-inequality}, 
\begin{align*}
\det \tilde{g} \leq ||T_V||^{2m}\det E = ||\adj(V)||^{2m}\det E.
\end{align*}
It follows that
\begin{align*}
\det g \leq \frac{||\adj(V)||^{2m}}{(\det V)^{2m}}\det E = \left(\frac{||\adj(V)||}{\det V}\right)^{2m} \det E= \left(\frac{\lambda_{\max}[\adj(V)]}{\det V}\right)^{2m}\det E.
\end{align*}
Let $\{\lambda_k\}_{k=1}^{2N}$ be the eigenvalues of $V$. From the relation $\adj(V) = \det(V)V^{-1}$, it follows that the eigenvalues of $\adj(V)$ are
$\left\{\frac{\det(V)}{\lambda_k}\right\}_{k=1}^{2N}$ and thus
\begin{align*}
\lambda_{\max}(\adj(V)) = \frac{\det(V)}{\lambda_{\min}(V)} \imply \det g \leq \left(\frac{1}{\lambda_{\min}(V)}\right)^{2m}\det E.
\end{align*}
This completes the proof of the theorem.
\end{proof}

{\begin{proof}[\textbf{Proof of Corollary \ref{corollary:integral-boundenergy}}]
The first expression of the Corollary follows from the bound given in Theorem \ref{theorem:det-bound} and the definition of $\varPhi$. 

We now show that the integral 
\begin{align*}
\int_{\Theta}\varPhi(V)\sqrt{\det g} d\theta,
\end{align*}
is bounded. By the inequality $\log(1+x) \leq x$ for all $x \geq 0$ and the limit $\lim_{x\approach 0}\frac{\log(1+x)}{x} = 1$, we always have 
\begin{align*}\
\frac{\log[1+(\det V)^m]}{(\det V)^m} \leq 1\;\;\; \text{whenever $\det V \geq 0$}. 
\end{align*}
Consider now the factor $H(\mathbf{E}-\tr(V))\lambda_{\max}^m(\adj(V))$.
As in the proof of Theorem \ref{theorem:det-bound}, let $\{\lambda_k\}_{k=1}^{2N}$ be the eigenvalues of $V$, arranged in decreasing order. Then
\begin{align*}
\lambda_{\max}(\adj(V)) = \frac{\det V}{\lambda_{2N}} = \prod_{j=1}^{2N-1}\lambda_j \leq \left(\frac{\sum_{j=1}^{2N-1}\lambda_j}{2N-1}\right)^{2N-1}
\leq \left(\frac{\trace(V)}{2N-1}\right)^{2N-1}\leq\left(\frac{\mathbf{E}}{2N-1}\right)^{2N-1},
\end{align*}
where $\mathbf{E}$ is a suitable \textit{real} positive constant bounding from above $\trace(V)$, which is a positive linear function in $\theta$. Furthermore, the domain of integration is now bounded from above, indeed it is given by $\left\{\theta\in\RR^m\ |\ V(\theta)>0\right\}\cap\left\{\theta\in\RR^m\ |\ \tr (V)\leq \mathbf{E} \right\}$. As sort of evidence, let us consider the following expression for $\det(V)$ \cite{Kondra},
\begin{align}\label{detV}
\det(V)=\sum_{k_1,\ldots,k_N}\prod_{l=1}^N\frac{(-1)^{k_l+1}}{l^{k_l}k_l!}\trace(V^l)^{k_l},
\end{align}
where the sum is taken over the set of all integers $k_l\geq 0$ satisfying the equation $\sum_{l=1}^N lk_l=N$. Now, whenever $V$ is positive definite matrix, through Cauchy-Schwarz inequality we have that $\trace(V^2)\leq\trace(V)^2$; thus, by induction on $n$, with $1\leq n\leq N$, we arrive at $\trace(V^l)\leq \trace(V)^l$ for every $l\leq N$. So, from Eq.\eqref{detV} we have that,
\begin{align}
\label{detVinequality}
0<\det(V)\leq \sum_{k_1,\ldots,k_N}\prod_{l=1}^N\frac{(-1)^{k_l+1}}{l^{k_l}k_l!}\mathbf{E}^{lk_l},
\end{align}
whenever $V$ is positive definite. Moreover, because of trace of a principal minor of $V$ is smaller than $\trace(V)$ then, we earn same bounds as in \eqref{detVinequality} for determinant of every principal minors whenever $V$ is positive definite. 
Thus, we have  the convergence of the integral. 
\end{proof}}


\section{Proof of Proposition \ref{regpro} and Corollary \ref{corollary:integral-bound} }\label{appendixD}

\begin{proof}[\textbf{Proof of Proposition \ref{regpro}}] Consider a permutation matrix $\Pi$ such that $V^\prime=\Pi^T V \Pi$. Then, because of the property of unitary determinant $\det \Pi =\det \Pi^T =1$, and the fact that the permutation matrices $\Pi$ do not depend on the
parameters $\theta^l$, we have
\begin{eqnarray*}
\det V^\prime &=& \det\left(\Pi^T V \Pi\right)=\det \Pi^T\ \det V \ \det \Pi\\
&=& \det V,\\
\mbox{tr}\left[(\det V^\prime)(V^\prime)^{-1} \right]&=& \det V^\prime\ \mbox{tr}\left[(V^\prime)^{-1} \right]\\
&=& \trace\left[(\det V)(V)^{-1} \right].
\end{eqnarray*}

Hence, from Eq. \eqref{reg} it immediately follows that $\Upsilon(V(\theta))=\Upsilon(V^\prime(\theta))$.

\bigskip

In the same way, since $\det S=1=\det S^T$, with $S$  a symplectic matrix  such that $V^\prime=S^T V S$, then we have
\begin{eqnarray*}
\det V^\prime &=&  \det V,\\
\mbox{tr}\left[(\det V^\prime)(V^\prime)^{-1} \right]&=& \mbox{tr}\left[(\det V)(V)^{-1} \right].
\end{eqnarray*}

Hence, from Eq. \eqref{reg} it immediately follows that $\Upsilon(V(\theta))=\Upsilon(V^\prime(\theta))$.  
\end{proof}

\begin{proof}[\textbf{Proof of Corollary \ref{corollary:integral-bound}}]
The first expression of the Corollary follows from the bound given in Theorem \ref{theorem:det-bound} and the definition of $\Upsilon$. 

We now show that the integral 
\begin{align*}
\int_{\Theta}\Upsilon(V)\sqrt{\det g} d\theta
\end{align*}
is bounded. By the inequality $\log(1+x) \leq x$ for all $x \geq 0$ and the limit $\lim_{x\approach 0}\frac{\log(1+x)}{x} = 1$, we always have 
\begin{align*}\
\frac{\log[1+(\det V)^m]}{(\det V)^m} \leq 1\;\;\; \text{whenever $\det V \geq 0$}. 
\end{align*}
Consider now the factor $\exp(-\trace[\adj(V)])\lambda_{\max}^m(\adj(V))$.
As in the proof of Theorem \ref{theorem:det-bound}, let $\{\lambda_k\}_{k=1}^{2N}$ be the eigenvalues of $V$, arranged in decreasing order. From the relation $\adj(V) = \det(V)V^{-1}$, it follows that the eigenvalues of 
$\adj(V)$ are $\left\{\frac{\det(V)}{\lambda_k}\right\}_{k=1}^{2N}$. We have
\begin{align*}
\trace[\adj(V)] &= \sum_{k=1}^{2N}\frac{\det V}{\lambda_k} 
= \sum_{k=1}^{2N}\prod_{j=1, j\neq k}^{2N}\lambda_j,
\end{align*}
which is a positive polynomial in the parameters $(\theta^i)_{i=1}^m$. Furthermore
\begin{align*}
\lambda_{\max}(\adj(V)) = \frac{\det V}{\lambda_{2N}} = \prod_{j=1}^{2N-1}\lambda_j \leq \left(\frac{\sum_{j=1}^{2N-1}\lambda_j}{2N-1}\right)^{2N-1}
\leq \left(\frac{\trace(V)}{2N-1}\right)^{2N-1},
\end{align*}
where $\trace(V)$ is a positive linear function in $\theta$.
Thus as $\theta$ grows, the expression
\begin{align*}
\exp(-\trace[\adj(V)])\lambda_{\max}^m(\adj(V))
\end{align*}
decays exponentially, leading to the convergence of the integral. 
\end{proof}



\end{document}